\documentclass[11pt]{article}

\usepackage[utf8]{inputenc}

\usepackage{amsthm,amsmath,amssymb}
\RequirePackage[numbers]{natbib}
\RequirePackage[colorlinks,citecolor=blue,urlcolor=blue]{hyperref}
\usepackage{array}
\usepackage{algorithm}
\usepackage{algorithmic}
\usepackage{color}
\usepackage{times}
\usepackage{graphicx}

\newtheorem{lemma}{Lemma}
\newtheorem{proposition}{Proposition}

\def\bo{\mathring{\beta}}
\def\b*{\mathring{\beta}}

\def\ml{\dot{\ell}}

\def\hbeta{\hat{\beta}}

\def\mA{\mathcal{A}}

\title{Group Lasso Merger For Sparse Prediction \\ 
With High-dimensional Categorical Data}

\author{Szymon Nowakowski \footnote{Center4ML, University of Warsaw, 
Banacha 2, 02-097 Warsaw, Poland, sd.nowakowski2@uw.edu.pl },
Piotr Pokarowski \footnote{Institute of Applied Mathematics and Mechanics,
University of Warsaw, Banacha 2, 02-097 Warsaw, Poland, 
pokar@mimuw.edu.pl }
,Wojciech Rejchel 
\footnote{Faculty of Mathematics and Computer Science,
Nicolaus Copernicus University,  Chopina 12/18, 87-100, Toru\'n, Poland, 
wrejchel@gmail.com } 
}

\date{}

\begin{document}

\maketitle

\begin{abstract}
Sparse prediction with categorical data is challenging even for a moderate number of variables, because one parameter is roughly needed to encode one category or level. The Group Lasso is a well known efficient algorithm for selection continuous or categorical variables, but all estimates related to a selected factor usually differ, so a fitted model may not be sparse. To make the Group Lasso solution sparse, we propose to merge levels of the selected factor, if a difference between its corresponding estimates is less than some predetermined threshold. We prove that under weak conditions our algorithm, called GLAMER for Group LAsso MERger, recovers the true, sparse linear or logistic model even for the high-dimensional scenario, that is when a number of parameters is greater than a learning sample size. To our knowledge, selection consistency has been proven many times for different algorithms fitting sparse models with categorical variables, but our result is the first for the high-dimensional scenario. Numerical experiments show the satisfactory performance of the GLAMER.
\end{abstract}

\section{Introduction}

Sparse predictive modelling in the presence of both numerical and categorical variables (factors) is challenging even for a moderate number of variables, because a factor with $k$ levels is usually encoded as $k-1$ dummy variables and $k-1$ parameters are needed to learn. Moreover, dimensionality reduction associated with numerical variables is conceptually simple (leave or delete), while for categorical predictors we can either exclude the whole factor or merge its levels. Thus,  feasible models consist of subsets of numerical variables and partitions of levels of factors. In this paper, the task of selecting such models will be referred to as ``partition selection''.

Sparse high-dimensional prediction, where a number of active variables is significantly smaller than a learning sample size $n$ and a number of all variables $p$ greatly exceeds $n,$ has been a focus of research in statistical machine learning in recent years. However, popular methods fitting sparse predictive models with high-dimensional data do not merge levels of factors: the Lasso \citep{Tibshirani96} treats dummy variables as separate, binary predictors, the Group Lasso \citep{groupLasso} can only leave or delete a whole factor and  the Sparse Group Lasso \citep{simon2013sparse} additionally removes levels of selected factors. 
The Fused Lasso \citep{fused2005} computes partition selection, but only in a simplified form for an ordered variable.  These methods do not realize partition selection, but they significantly reduce a number of parameters and select
 variables, which may be an input for further, interpretable dimension reduction techniques.

In the mainstream research on the Lasso-type algorithms, the CAS-ANOVA method \citep{bondell2009} fits sparse linear models with fusion of factor levels using the $l_1$ penalty imposed on difference between parameters corresponding to levels of a factor. CAS-ANOVA has been implemented several times and now fits Generalized Linear Models \citep{gertheiss2010sparse, oelker14}. An alternative to the penalization is a greedy search as in the DMR algorithm for linear or logistic models \citep{MajK2015}. The method computes a nested family of candidate models via agglomerative clustering with a dissimilarity measure given by  likelihood ratio statistics corresponding to deletion of one variable or merging two levels of a factor. Finally, a model is selected by minimization of an information criterion. 

Recently, a growing interest in partition selection has been noticed. \citet{pauger2019bayesian} introduced a Bayesian method for linear models based on a prior inducing fusion of levels. Another Bayesian procedure is proposed by \citet{garcia2021variable}. The most recent algorithm for partition selection is probably SCOPE \citep{stokell2021}. The method uses a minimax concave penalty on differences between consecutive, sorted estimators of coefficients for levels of a factor. SCOPE fits linear or logistic models. Let us note that all above-mentioned methods for partition selection are restricted to a classical scenario $p<n$, except SCOPE and DMR, which in its new implementation is based on variables screened by the Group Lasso \citep{DMRnet}. 

The main contributions of this paper are summarized as follows:

1. We propose a new partition selection method for linear or logistic models. Our algorithm,  called GLAMER for  GRoup LAsso MERger, consists of of three steps. First, it computes the Group Lasso estimator and sorts the coefficients for each factor. In the second step it merges levels corresponding to subsequent estimates, if the difference between them is less than a predetermined threshold. In the third step an estimator is refitted by the maximum likelihood to reduce the estimation bias and improve prediction. GLAMER is a generalization of the Thresholded Lasso algorithm \citep{Zhou09, BuhlmannGeer11} for subset selection. It easily improves the Group Lasso's ability to choose a sparse model.

2. We prove an upper bound for  $l_{\infty}$ estimation error of the Group Lasso with an additional 
 diagonal matrix of weights.  Next, for an orthogonal design, we minimize the obtained bound with respect to the weights.   Our optimal weights are different  from those recommended by \citet{groupLasso}. 

3. We prove under weak conditions that GLAMER recovers the true, sparse linear or logistic model even for $p>>n$. To our knowledge, selection consistency has been proven many times for different algorithms fitting sparse models with categorical variables, but our result is the first for the high-dimensional scenario. We show that for an orthogonal design  our sufficient condition for partition selection is also necessary up to a universal constant.

4. In theoretical considerations the Lasso-type algorithms are defined for one penalty and return one coefficient estimator. However, practical implementations usually use nets of data-driven penalties and return lists of estimators. Our next contribution is  an analogous implementation of the GLAMER algorithm. In numerical experiments on four real data sets we compare GLAMER to competitive methods and show the satisfactory performance of our algorithm.

In the rest of this paper we describe the considered models and the GLAMER algorithm. We also present mathematical propositions with proofs, which describe properties of our method. Finally, we compare the GLAMER to other methods for sparse prediction in numerical experiments.

\section{Models and the algorithm}

The way we model data will encompass normal linear and logistic models as premier examples.
We consider independent data $(y_1,x_{1.}),(y_2,x_{2.}),\ldots,(y_n,x_{n.})$, 
where $y_i \in \mathbb{R}$ is a response variable and $x_{i.} \in\mathbb{R}^p$ is a vector of predictors. Every vector of predictors $x_{i.}$ can consist of continuous predictors as well as categorical predictors. 
We arrange them in the following way
$x_{i.}=(1, x_{i1}^T,x_{i2}^T,\ldots,x_{ir}^T).^T
$
Suppose that   $x_{ik}$ corresponds to a categorical predictor (factor) for some $k \in \{1,\ldots,r\}.$ Then a set of  levels of this factor is given by $\{0,1,2,\ldots,p_k\}$ and $x_{ik} \in \{0,1\}^{p_k}$ is a dummy vector corresponding to $k$-th predictor of $i$-th object in a data set. So, a reference level, say the zero level, is not included in $x_{ik}.$
If $x_{ik}$ corresponds to a continuous predictor, then simply $x_{ik} \in \mathbb{R}^{p_k}$ and $p_k=1.$
Therefore, a dimension of $x_{i.}$ is   $p=1+\sum_{k=1}^r p_k.$ Finally, let $X=[x_{1.},\ldots,x_{n.}]^T$ be a $n\times p$ design matrix.

We assume that for some $\mathring{\beta} \in \mathbb{R}^p$ and a known, convex and differentiable  function 
$\gamma:\mathbb{R}\rightarrow\mathbb{R}$ 
\begin{equation}
\label{inverseLink}
 \mathbb{E}y_i=\dot\gamma(x_{i.}^T\mathring{\beta})~\text{for}~i=1,2,\dots,n,
\end{equation}
where $\dot \gamma$ denotes the derivative of $\gamma.$
Coordinates of $\bo$ correspond to coordinates of a vector of predictors, that is
$
\bo=(\bo _{0}, \bo _1^T, \bo _2^T, \ldots, \bo _r^T
)^T,
$
where $\bo _{0} \in \mathbb{R}$ relates to an intercept and $\bo _k = (\bo _{1,k}, \bo _{2,k}, \bo _{3,k}, \ldots,
\bo _{p_k,k} )^T\in \mathbb{R}^{p_k}$ for  $k=1,\ldots, r.$ 
Note that (\ref{inverseLink}) is satisfied in particular by the Generalized Linear Models (GLMs) 
with a canonical link function
and a nonlinear regression with an additive error.

Moreover, we suppose that centred responses $\varepsilon_i = y_i - \mathbb{E}y_i$ have a {\it subgaussian distribution}
with the same  number $\sigma>0$, that is for $i=1,2,\dots,n$ and $u \in \mathbb{R}$ we have
\begin{equation}
\label{subgauss}
\mathbb{E}\exp(u\varepsilon_i) \leq \exp(\sigma^2u^2/2).
\end{equation}

\noindent {\bf Examples.} Two most important representatives of GLMs are the normal linear model and logistic regression. In the normal linear model 
$$
y_i=x_{i.}^T\mathring{\beta} + \varepsilon _i, \quad i=1,2,\dots,n,
$$
where the noise variables $\varepsilon_i$ are independent and normally distributed $N(0,\mathring{\sigma}^2).$ Therefore,  assumptions \eqref{inverseLink} and \eqref{subgauss}  are satisfied with   $\gamma(a)=a^2/2 $ and 
 any $\sigma\geq \mathring{\sigma},$ respectively. In logistic regression the response variable 
is dichotomous $y_i
\in \{0,1\}$ and we assume that 
$$
P(y_i=1) = \frac{\exp(x_{i.}^T\mathring{\beta})}{\exp(x_{i.}^T\mathring{\beta})+1}, \quad i=1,\ldots,n.
$$
In this model assumption \eqref{inverseLink} is satisfied with   $\gamma(a)=\log(1+\exp(a)).$ Finally, as $(\varepsilon_i)$ are bounded random variables,  then (\ref{subgauss}) is satisfied with any $\sigma\geq 1/2$.

\subsection{Notations}

Let  $W_k=diag(w_{1,k},\ldots, w_{p_k,k}),k=1,\ldots,r$ be diagonal nonrandom matrices with positive entries. Besides, $W=diag(W_1,\ldots, W_r)$ is a $p \times p$ diagonal matrix with matrices $W_k$ on the diagonal. Next, for $\beta \in \mathbb{R}^p$ and $q\geq 1$  
let $|\beta|_q = (\sum_{j=1}^p |\beta_j|^q)^{1/q}$ be the $\ell_q$ norm of $\beta$. The only  exception is  the $\ell _2$ norm, for which we will use the special notation $||\beta||.$

A feasible model is defined as a sequence $M=(P_1,P_2,\ldots,P_r).$ If $k$-th predictor is a factor, then $P_k$ is a particular partition of its levels. If $k$-th predictor is continuous, then $P_k \in \{\emptyset,\{k\}\}.$
Every $\beta \in \mathbb{R}^p$ determines a model $M_\beta$ as follows: if $k$-th predictor is a factor, then partition  $P_k$ depends on the fact, whether equalities of the form $\beta_{j,k}=0$ and/or $\beta_{j_1,k}=\beta_{j_2,k}$ are satisfied, which correspond to merging $j$-th level with a reference level and merging levels $j_1$ and $j_2,$ respectively.  If $k$-th predictor is continuous, then $P_k=\{k\}$ when $\beta_k \neq 0$ and $P_k=\emptyset$ otherwise.

In the following we consider $k \in \{1,\ldots, r\}$ and  $j \in \{1,\ldots,p_k\} .$  Let  $x_{j,k}$ be a column of $X$ corresponding to $j$-th level of $k$-th factor. The additional notations are $x_M=\max_{j,k} \; ||x_{j,k}||, x_m=\min_{j,k} \; ||x_{j,k}||, x_W=\max_{j,k}\; ||x_{j,k}||/w_{j,k}.$
Finally, $\Delta= \min\limits_{1\leq k\leq r} \;\;\min\limits _{0\leq j_1,j_2\leq p_k: \bo _{j_1,k} \neq \bo _{j_2,k}} |\bo _{j_1,k} - \bo _{j_2,k}|,$ where we set $\bo _{0,k}=0$ for $k=1,\ldots ,r.$

\subsection{The algorithm}
For estimation of  $\mathring{\beta}$ we consider a {\it loss function}
\begin{equation}
\label{loss}
\mathcal{L} (y, X\beta) \equiv \ell(\beta)=\sum_{i=1}^n [\gamma(x_{i.}^T\beta)-y_ix_{i.}^T\beta],
\end{equation}
which is a negative log-likelihood.
It is easy to see that $\dot \ell(\beta) =\sum_{i=1}^n [\dot \gamma(x_{i.}^T\beta)-y_i]x_{i.}$,
and consequently $\dot \ell(\mathring{\beta}) = -X^T\varepsilon$ for $\varepsilon=(\varepsilon_1,\ldots,\varepsilon_n)^T$. Next,
for $k=1,\ldots, r$ partial derivatives of $\ell (\beta)$ corresponding to coordinates of $\beta_k$ are denoted by $\dot \ell _k(\beta).$ 

We present GLAMER in Algorithm~\ref{alg:glamer}. It consists of three steps: the first one is the Group Lasso, where diagonal elements $(W_k)_{jj}=||x_{j,k}||$ play roles of weights. Such a choice of weights is explained in Proposition \ref{opt_w}. The second step concerns merging levels of each categorical predictor separately (for continuous predictors we just discard them or not). It is done using Group Lasso coefficients and threshold $\tau.$ More precisely, consider $k$-th categorical predictor. First, every level such that its corresponding Group Lasso coefficient is smaller than $\tau$ is merged with a reference level. Then Group Lasso coefficients corresponding to the remaining levels are sorted. Now merging levels into clusters is simple and follows the proof of Proposition \ref{prop_consist}. There  we establish that two Group Lasso coefficients corresponding to levels from the same cluster  do not differ more than  $2\tau,$ while those corresponding to distinct clusters have to differ more than   $2 \tau.$ 
Finally, we construct a new design matrix $Z$ in the following way: for every factor we add those columns of initial matrix $X,$ which correspond to the same clusters.  Then we calculate a maximum likelihood estimator using matrix $Z$.

\begin{algorithm}[tb]
   \caption{GLAMER}
   \label{alg:glamer}
 \textbf{Input}: $y, X, \lambda, \tau,  W_k$ for $k=1,\ldots,r,$ where by default $(W_k)_{jj}=||x_{j,k}||$ for $j=1,\ldots,p_k.$
\begin{algorithmic}
\STATE {\bfseries Step 1: Group Lasso}  
 \STATE   {$~~\widehat{\beta} = \text{argmin}_{\beta} \left\{ \mathcal{L} (y, X\beta) + \lambda 
\sum_{k=1}^r ||W_k \beta_k|| 
 \right\} $};
 \STATE {\bfseries Step 2: Merging}
\STATE {Start with the $n$-dimensional vector $Z$ of ones};
\STATE{{\bf for} $k=1$ {\bf to} $r$ {\bf do}}
\STATE{\quad Sort levels:
$\widehat \beta_{j_1,k} \leq
\widehat \beta_{j_2,k} \leq \ldots \leq \widehat \beta_{j_{p_k},k} ;  $
}
\STATE{\quad Merge levels into clusters: let $m=1$ and $C_m=\{j_1\}$
}
\STATE{\quad \quad {\bf for} $i=2$ {\bf to} $p_k$ {\bf do}
}
\STATE{\quad \quad  \quad {\bf if} $\widehat \beta _{j_i,k} - \widehat \beta _{j_{i-1},k} \leq  \tau$ {\bf then }
$C_m =C_m \cup \{j_i\}$ 
}
\STATE{ \quad \quad \quad {\bf else} $m=m+1$, $C_m=\{j_i\}$
}
\STATE{\quad \quad\bf{end for}}
\STATE{\quad Add merged columns of $X_k$ to $Z$:}
\STATE{\quad \quad {\bf for} $i=1$ {\bf to} $m$ {\bf do}}
\STATE{\quad \quad \quad $Z=[Z,\sum_{j \in C_i} X_{jk} ] $}
\STATE{\quad \quad\bf{end for}}
\STATE{\bf{end for}}
\STATE {\bfseries Step 3: Maximum likelihood estimation  }
\STATE {$\widehat{\beta}_{GLAMER} = \arg \min_\beta \mathcal{L} (y, Z\beta)$}
\end{algorithmic}
 \textbf{Output}:  $\widehat{\beta}_{GLAMER}, \widehat{M}_{GLAMER}=M_{\widehat{\beta}_{GLAMER}}$
\end{algorithm}

\section{Sufficient and necessary conditions for selection consistency of GLAMER}

We consider the GLAMER algorithm with arbitrary diagonal matrices $W_k.$ The default setting in Algorithm \ref{alg:glamer} will be justified in Proposition \ref{opt_w}.

First, we generalize a characteristic of linear models with continuous predictors, which quantifies the degree of separation
between  model $M_{\bo}$ and other models \citep{YeZhang10}.

Let $S=\{1 \leq k \leq r: \bo _k \neq 0\}$  and $\bar S = \{1, \ldots, r\} \setminus S$. 
Notice that $S$ need not coincide with $M_{\bo}.$
For $ a \in (0,1)$ and a diagonal matrix $W$ we define a cone 
\begin{eqnarray}
 \label{cone}
{\cal C}_{a,W}=\{v \in \mathbb{R}^p:  
\sum_{k \in \bar S} ||W_k v_k||  \leq  
\sum_{k \in S}  ||W_k v_k|| +  a |W v|_1\}.
\end{eqnarray} 
A Cone Invertibility Factor (CIF)  is defined as
 \begin{equation}
 \label{CIF}
 \zeta_{a,W}=\inf_{0 \neq \nu\in {\cal C}_{a,W}}\frac {\left|W^{-1}\left[\dot{\ell}(\bo + \nu) - \dot{\ell}(\bo )\right]\right|_\infty}
{|\nu|_\infty}\:.                                         
 \end{equation}
Notice that in the linear model the numerator of \eqref{CIF} is simply  $|W^{-1}X^TXv|_\infty.$ 
If all predictors are continuous, then \eqref{cone} and {CIF} are the same as a cone and CIF in \citet{YeZhang10}.

In the case $n>p$ one usually uses the minimal eigenvalue of the matrix $X ^TX$ to express the strength of correlations between predictors. Obviously, in the high-dimensional scenario this value is zero. 
Therefore, CIF can be viewed as a useful analog of the minimal eigenvalue for the case $p>n.$
In comparison to more popular restricted eigenvalues \citep{BickelEtAl09} or compatibility constants \citep{GeerBuhlmann09}, CIF enables sharper 
$\ell_\infty$ estimation error bounds    \citep{YeZhang10, HuangZhang12, ZhangZhang12}. 

In the following result we investigate an estimation error of the Group Lasso.

\begin{lemma}
\label{lem1}
Suppose that assumptions \eqref{inverseLink}, \eqref{subgauss} are satisfied and  $a \in (0,1).$  
Then
 \begin{equation*}
 \mathbb{P_{\bo}} \left(|\hat{\beta}-\mathring{\beta}|_\infty > (1+a)\lambda\zeta_{a,W}^{-1} \right) \leq 2p\exp\bigg(-\frac{a^2\lambda^2}{2\sigma^2 x^2_W}
 \bigg).
 \end{equation*}
\end{lemma}

\begin{proof}
For $k=1,\ldots, r$ using KKT for  the Group Lasso estimator $\hat{\beta},$ we have that $W_k^{-1} \ml _k (\hat{\beta})= 
 -\lambda W_k \hbeta _{k}/||W_k \hbeta _k|| $ for $\hbeta _k \neq 0$ and $||W_k^{-1} \ml _k (\hat{\beta})|| \leq
 \lambda $ for $\hbeta _k = 0.$ Therefore, we obtain $|W^{-1} \ml  (\hat{\beta})|_\infty= \max _{k} |W_k^{-1} \ml _k (\hat{\beta})|_\infty \leq \lambda.$

Recall that $\ml (\bo)=-X^T\varepsilon$ and
suppose that we are on  event $\mathcal{A}=\{|W^{-1}\ml (\bo)|_\infty\leq a\lambda\}.$
First,  we  prove that $v:=\hat \beta - \bo \in {\cal C}_{a,W}.$ Using convexity and differentiability of $\ell,$ from Taylor's expanssion we have $0 \leq v ^T \left[  \dot \ell (\hbeta) - \dot \ell  (\bo)\right].$
Since $v _k =\hbeta _k$ for $k \in \bar S,$ we obtain
\begin{eqnarray}
\label{form1}
0 &\leq& v ^T \left[  \dot \ell (\hbeta) - \dot \ell  (\bo)\right]=
\sum_{k=1}^r v _k ^T \dot \ell _k (\hbeta) - \sum_{k=1}^r v _k ^T \dot \ell _k (\bo) \nonumber \\ 
&=& \sum_{k \in \bar S} \hbeta _k ^T \dot \ell _k (\hbeta) + \sum_{k \in S} v _k ^T \dot \ell _k (\hbeta) 
 - \sum_{k=1}^r v _k ^T \dot \ell _k (\bo).
\end{eqnarray}
Consider the first term in \eqref{form1}. Using KKT, it equals
\begin{eqnarray*}
 \sum_{k \in \bar S,\hbeta _{k} \neq 0} 
 \hbeta _{k} ^T \dot \ell _{k} (\hbeta)
= -\lambda \sum_{k \in \bar S,\hbeta _{k} \neq 0}  \;||W_k\hbeta _k||  =
-\lambda \sum_{k \in \bar S}  \;||W_kv _k||.
\end{eqnarray*}
Similarly, we bound the second term in \eqref{form1} by
\begin{eqnarray*}
 \sum_{k \in  S} 
 ||W_k v _k||\; || W_k ^{-1}\dot \ell _{k} (\hbeta)||  \leq \lambda \sum_{k \in  S}||W_k v _k||.
\end{eqnarray*}
The last term in \eqref{form1} can be bounded using the fact that we are on  event $\mathcal{A}$ 
$$
\sum_{k=1}^r |W_k v _k|_1 |W_k^{-1} \dot \ell _k (\bo)|_\infty 
\leq a\lambda \sum_{k =1}^r |W_k v _k|_1.
$$
Joining the above facts we get that $v \in  {\cal C}_{a,W}.$
Therefore, from the definition \eqref{CIF} we have
\begin{eqnarray*}
&\,& \zeta_{a,W}|\hbeta - \bo|_\infty \leq 
\max\limits_{1\leq k \leq r} |W_k^{-1}\dot \ell _k (\hbeta)
-W_k ^{-1} \dot \ell _k (\mathring{\beta})|_\infty  \\
&\,&\leq \max_{1\leq k \leq r} 
 |W_k ^{-1}\dot \ell _k (\hbeta)|_\infty
+\max_{1\leq k \leq r} 
|W_k^{-1} \dot \ell _k (\mathring{\beta})|_\infty .
\end{eqnarray*}
Using again KKT and the fact, that we are on $\mA ,$ we get $|\hat{\beta}-\mathring{\beta}|_\infty \leq (1+a)\lambda\zeta_{a,W}^{-1}.$ Now we calculate probability of event~$\mA.$
To do it, we use the following exponential inequality for independent subgaussian variables $\varepsilon_i, i=1,\ldots ,n$: for each $b>0$ and $v\in \mathbb{R}^n$ we have $P(\varepsilon^Tv/||v||>b) \leq \exp\left(-b^2/(2\sigma^2)\right).$ 
Using union bounds and the definition of $x_W,$ we obtain
\begin{eqnarray*}
&&P_{\bo}(\mA ^c) \leq \sum_{k,j} P\left(|x_{j,k}^T \varepsilon|/w_{j,k} > a\lambda \right) \\
&&\leq 2 \sum_{j,k} \exp \left(
-\frac{a^2 \lambda^2 w_{j,k}^2}{2 \sigma^2 ||x_{j,k}||^2} 
\right) 
 \leq
2p  \exp \left(
-\frac{a^2 \lambda^2 }{2 \sigma^2 x_W^2} \right),
\end{eqnarray*}
where we consider $k \in \{1,\ldots, r\}$ and  $j \in \{1,\ldots,p_k\}. $ 
\end{proof}

The next fact is a simple consequence of Lemma \ref{lem1}.

\begin{proposition}
\label{prop_est}
Under assumptions of Lemma \ref{lem1} and $\lambda ^2= 2 a^{-2}\sigma^2x^2_W \log(2p/\alpha)$ for some $\alpha \in (0,1)$ we have that with probability at least $1-\alpha$
$$
|\hat{\beta}-\mathring{\beta}|^2_\infty/\sigma^2 \leq 2(1+a)^2a^{-2} x_W^2 \zeta_{a,W}^{-2} \log(2p/\alpha). 
$$
\end{proposition}

In Lemma \ref{lem1} and Proposition \ref{prop_est} we establish an upper bound on the estimation error of the Group Lasso. It works for subgaussian GLMs and the high-dimensional scenario  $p>>n.$ Similar results can be found in the literature, for instance in \citet[Theorem 8.1]{BuhlmannGeer11}, \citet[Theorem III.6]{Blazere2014} or \citet[Theorem 5.1]{Lounici_etal2011}. The main difference between those results and ours is that we measure the estimation error in the $l_\infty$ norm, which is all we need to prove selection consistency, while in those papers we have a mixture of $l_2$ and $l_1$ (or $l_\infty$) norms. Hence their bounds contain unnecessary terms, which makes their assumptions more restrictive, for instance
  a sum of levels of categorical predictors corresponding to $M_{\bo}.$ Moreover, \citet{BuhlmannGeer11} use the same weights as in the original version of the Group Lasso \citep{groupLasso}. As we explain after Proposition \ref{opt_w}, such weights lead to sub-optimal bounds.  

\subsection{Sufficient conditions for selection consistency}

\begin{proposition}
\label{prop_consist}
Under assumptions of Proposition \ref{prop_est}, if $\alpha=o(1)$ and 
\begin{equation}
\label{ass_delta}
8(1+a)^2a^{-2} x_W ^2\zeta_{a,W}^{-2} \log p (1+o(1))\leq  \tau ^2/\sigma^2 < \Delta^2 /(4\sigma^2)
\end{equation}
 we have that $P_{\bo} (\hat M _{GLAMER} \neq M_{\bo})=o(1).$
\end{proposition}

\begin{proof}
From Proposition \ref{prop_est} and \eqref{ass_delta} we know that $|\hat{\beta}-\mathring{\beta}|_\infty \leq \tau/2 < \Delta/4.$ Now we fix the $k$-th predictor and take indexes $j_1,j_2$ such that $\bo _{j_1,k}= \bo _{j_2,k},$ i.e. they correspond to the same cluster. We obtain
\begin{equation}
\label{sep1}|\hbeta _ {j_1,k} - \hbeta_{j_2,k}|\leq |\hbeta _ {j_1,k} - \bo _{j_1,k}| + |\hbeta _ {j_2,k} - \bo _{j_2,k}| 
\leq \tau.
\end{equation}
On the other hand, if $j_1,j_2$ are such that
$\bo _{j_1,k} \neq \bo _{j_2,k}$, then
\begin{eqnarray}\label{sep2}
&&|\hbeta _ {j_1,k} - \hbeta_{j_2,k}|  
\geq |\bo _{j_1,k} -\bo _{j_2,k}| - |\hbeta _ {j_1,k} - \bo _{j_1,k}| - |\hbeta _ {j_2,k} - \bo _{j_2,k}|\nonumber\\
&&\geq \Delta  -\tau >\tau.
\end{eqnarray}
Inequalities \eqref{sep1} and \eqref{sep2} state that Algorithm \ref{alg:glamer} correctly merges levels into clusters.  
\end{proof}

For $p \leq n$  consistency of partition selection was obtained in \citet{bondell2009, gertheiss2010sparse, oelker14, MajK2015} or \citet{stokell2021}.  To find out how restrictive the assumptions in Proposition \ref{prop_consist} are, we will later compare them to necessary conditions for selection consistency, which are given in  Proposition \ref{prop_necess}. 

To our knowledge, Proposition \ref{prop_consist} is the first result, which establishes consistency in partition selection  for  high-dimensional GLMs with categorical predictors.
  There are some papers, which prove ``group'' selection consistency for $p>>n$, because selection is understood as finding relevant groups of predictors, for instance \citet{NardiRinaldo2008, Lounici_etal2011, Blazere2014}. To be more precise, they are able only to find set $S$ instead of $M_{\bo},$ because the Group Lasso does not merge levels of factors.

We can write \eqref{ass_delta} in the simplified form $\Delta^2/\sigma^2 \succsim 128 x^2_W \zeta_{a,W}^{-2} \log p.$ Therefore, the estimation error in Proposition \ref{prop_est} and 
the sufficient conditions for selection consistency of GLAMER depend on a choice of weights.  We see that  to find optimal weights we should minimize $ x^2_W \zeta_{a,W}^{-2}.$ Solving this problem in the general case is difficult, so we restrict to the simplified version of the problem in the next result.
\begin{proposition}
\label{opt_w}
Consider  a linear model with an orthogonal design, i.e. $X^TX$ is orthogonal, and weights of the form $w_{j,k}=||x_{j,k}||^q$ for  $q \in \mathbb{R}.$ Then for each $a \in (0,1)$ we have
$$
x^2_W \zeta_{a,W}^{-2}\leq x^{-2}_m (x_M/x_m)^{\max(0,|2q-3|-1)}=:f(q)
$$
and $\arg \min_q f(q)=[1,2].$
\end{proposition}

\begin{proof}
For a linear model with an orthogonal design we can easily bound from above $
\zeta_{a,W}^{-2}
$ by $x_m^{2q-4},$ when $q\leq 2$ and $x_M^{2q-4},$ when $q> 2.$  The rest of the proof follows from the fact that $x^2_W$ equals $x_M^{2-2q},$ when $q\leq 1$ and $x_m^{2-2q},$ when $q> 1.$
\end{proof}
Thus, for an orthogonal design with the optimal weights the sufficient condition for selection consistency of GLAMER is $\Delta^2/\sigma^2 \succsim 128 x^{-2}_m \log p.$
The assumption that a design is orthogonal is quite restrictive. The much more common case, especially for $p>>n,$ is an {\it almost orthogonal} design, i.e. ${x_{j_1,k_1}}^Tx_{j_2,k_2}=o(x_m^2)$ for $(j_1,k_1) \neq (j_2,k_2).$ In such a case weights 
$w_{j,k}=||x_{j,k}||^q$ for  $q \in [1,2]$ can be treated as {\it almost optimal}.

Consider a linear model with only one categorical predictor. It is a simple example of an orthogonal design. It is also a case that $p \leq n,$ so we can compare  results from \citet{MajK2015, stokell2021} to ours.
For the choice of weights as in Proposition \ref{opt_w} the sufficient condition for selection consistency of GLAMER is $\Delta^2/\sigma^2 \succsim 128 x^{-2}_m  \log p.$ Results obtained in those papers differ only in constants, namely one has 32 in \citet[Theorem 1]{MajK2015}  and 420 in \citet[Theorem 5]{stokell2021}.

Finally, we discuss commonly used weights in the Group Lasso penalty. In the original paper on the Group Lasso \citep{groupLasso} two choices of weights are proposed. The first one, called ``obvious'', gives a penalty 
of the form $\lambda \sum_k ||\beta_k||.$ In the second one, called ``preferred'' they have a penalty  $\lambda \sum_k \sqrt{p_k} ||\beta_k||.$ The latter choice is more widely used in the literature \citep{BuhlmannGeer11}.
Now we compare these choices of weights to those obtained in Proposition \ref{opt_w}. Notice that columns of $X$ are normalized in \citet{groupLasso}, which is not done in our paper. So, we start with writing their penalty in our setting.
We do it under a {\it balanced design}, i.e. there are $n/p_k$ observations on every level of $k$-th categorical predictor. Their first choice gives a penalty $ \lambda \sqrt{n} \sum_k p_k^{-1/2} ||\beta_k||,$ while the second one gives $ \lambda \sqrt{n} \sum_k  ||\beta_k||.$ On the other hand, by Proposition \ref{opt_w} for $q=1$ we obtain
$ \lambda \sqrt{n} \sum_k p_k^{-1/2} ||\beta_k||,$ while for $q=2$ we have $ \lambda n \sum_k p_k^{-1} ||\beta_k||.$
Therefore, our optimal choice for $q=1$ coincides with the ``obvious'' choice in \citet{groupLasso}. However, the ``preferred'' choice in \citet{groupLasso}, which is widely used in practice, leads to sub-optimal results. Obviously, 
Proposition \ref{opt_w} deals with an orthogonal design, so our result is rather a starting point of the thorough analysis on weights optimality. 

\subsection{Necessary  conditions for selection consistency}

First, we need a few additional definitions. We omit $\bo _{0,1} $ in $\bo _1,$ which means that we consider a model without an intercept. Then we consider the following set of vectors $B(\bo)=\{\bo,\beta^{(j,k)}: k=1,\ldots,r,j=1,\ldots,p_k\}$, where vectors $\beta^{(j,k)} \in \mathbb{R}^p$ are defined as follows:  $\beta^{(j,k)}_{j_1,k_1} = 
\bo _{j_1,k_1} + \Delta I(j_1=j,k_1=k)$ for $k,k_1=1,\ldots,r, j=1,\ldots,p_k, j_1=1,\ldots,p_{k_1}.$ Therefore, vector $\beta^{(j,k)} $ differs minimally  from $\bo$ only on one coordinate, which corresponds to indexes $(j,k).$ 
\begin{proposition}
\label{prop_necess}
Consider the linear model with $\varepsilon \sim N(0,\sigma^2 I_n).$ Then for every selector $\hat M$ we have: if $\max_{\beta \in B(\bo)} P_\beta( \hat M \neq M_\beta)=o(1),$ then $ \Delta^2/\sigma^2 \geq x^{-2}_M \log p (1-o(1))/2.$ 
\end{proposition}
\begin{proof}
The main ingredient of the proof is Fano's inequality \citep{Ibragimov81, ShenEtAl13}. To apply it we need to bound the Kullback-Leibler divergence $D(j_1,k_1||j_2,k_2)$ between any pair of distributions $P_{\beta^{(j_1,k_1)}}$
 and $P_{\beta^{(j_2,k_2)}}.$ Let $\mathring \mu=X \bo $ and $\mu_{jk} = X\beta^{(j,k)}.$ Using gaussianity of $\varepsilon,$ we obtain for $(j_1,k_1)\neq (j_2,k_2)$ that 
\begin{eqnarray*}
&\,&2\sigma^2 D(j_1,k_1||j_2,k_2) =||  \mu_{j_1,k_1} -  \mu_{j_2,k_2} ||^2\\
&\,& \leq 2 ||  \mu_{j_1,k_1} - \mathring \mu ||^2 + 2 ||  \mu_{j_2,k_2} - \mathring \mu ||^2 \leq 4 x^2_M\Delta^2.
\end{eqnarray*}
Therefore, applying Fano's inequality we obtain
$$
2 x^2_M \Delta^2/\sigma^2 \geq \log p \left(1- \max_{\beta \in B(\bo)} P_\beta( \hat M \neq M_\beta) -\frac{\log 2}{ \log p}
\right).
$$
\end{proof}

Thus, the necessary condition for selection consistency is roughly $ \Delta^2/\sigma^2 \succsim x_M^{-2} \log p/2.$ Recall that for GLAMER  the sufficient condition is $ \Delta^2/\sigma^2 \succsim 128 x_m^{-2} \log p,$ if a design is orthogonal. If we suppose that  $x_m \approx x_M$ (i.e. minimal and maximal numbers of observations per factor level are similar),   then the sufficient condition for GLAMER coincides with the necessary condition with respect to the constant.

\section{Experiments}

In the theoretical analysis of Lasso-type estimators one usually considers only one value of tuning parameter $\lambda.$ We have also followed this way. However, the practical implementations can efficiently return   estimators for a data driven net of tuning parameters, as in the R~package {\tt glmnet}  \cite{FriedmanEtAl10}. Similarly, using a net of $\lambda$'s, the Group Lasso and the Group MCP algorithms
have been implemented in the R~package {\tt grpreg} \citep{BrehenyHuang15}. In the paper we also propose
a net modification of the GLAMER algorithm.

While working on the GLAMER implementation, we noticed that it is similar to the DMR implementation for high dimensional data, so we have decided to use the procedures from the R~package {\tt DMRnet} \citep{DMRnet}. Our implementation is based on two observations. First, from  \citet[Lemma 6]{MajK2015} and condition $ |\hat \beta - \mathring \beta |_{\infty} < \Delta/4 $ we obtain that the "reasonable" agglomerative clustering will have the true model $M_{\mathring \beta}$ on the dendrogram. Hence, the second step of GLAMER can be replaced, for example, by the complete linkage algorithm. Second, when designing a grid of hyperparameters, it is worth to replicate not only $\lambda$, but also $\tau$. However, it is easy to see that for enough $\tau$ values, the result of the procedure will be the entire dendrogram, so we can exclude $\tau$ from the implementation and to process the dendrograms, which are families of nested models. In this way, we obtain the GLAMER implementation scheme.

1. For $\lambda$ belonging to the grid:

   (i) calculate the Group Lasso estimator $\hat \beta(\lambda)$,

   (ii) perform complete linkage for each factor and  get nested family of models $M_1(\lambda)\subset M_2(\lambda)\subset \ldots$

2. For a fixed model dimension $c$, select a model $ M_c$ from  family  $(M_c(\lambda))_\lambda$, which has the minimal prediction loss.

3. Select a final model from sequence $(M_c)_c$ using cross-validation or  the Risk Inflation Criterion (RIC), see \citet{FosterGeorge1994}.

GLAMER uses the Group Lasso estimator $\hat \beta$ to compute dissimilarity measure, but DMR refits $\hat \beta$ with maximum likelihood and computes dissimilarity as the likelihood ratio statistics.

\subsection{Data sets with binary responses}

Adult data set~\citep{Kohavi} contains data from the 1994 US census. It contains 32,561 observations in a file \texttt{adult.data} and 16,281 observations in a file \texttt{adult.test}. The response  represents whether the individual's income is higher than 50,000 USD per year or not. We preprocessed the data as in~\citet{stokell2021}, i.e. we combined  two files together, removed 4 variables representing either irrelevant (\emph{fnlwgt}) or redundant (\emph{education-num}) features or with values for the most part equal to zero (\emph{capital-gain} and \emph{capital-loss}) and then removed the observations with missing values. Preprocessing resulted in 45,222 observations with 2 continuous and 8 categorical variables with $p=93$.

Promoter data set~\cite{Harley, Towell} contains E. Coli genetic sequences of length 57. The response  represents whether the region represents a gene promoter. We removed the \emph{name} variable and further worked with a data set consisting of 106 observations with 57 categorical variables, each with 4 levels representing 4 nucleotides, thus with $p=172$.

Both data sets are available at the UCI Machine Learning Repository~\cite{Dua_Graff}.

\subsection{Data sets with continuous responses}

Insurance data set~\cite{Kaggle} contains data describing attributes of life insurance applicants. The response is an 8-level ordinal variable measuring insurance risk of the applicant, which we treat as a continuous response. We preprocessed the data as in~\citet{stokell2021}, i.e. we removed the irrelevant \emph{id} variable and 13 variables with missing values. Preprocessing resulted in 59,381 observations with 5 continuous and 108 categorical variables with $p=823$.

Antigua data set~\cite{Andrews} contains data concerning maize
fertilizer experiments on the Island of Antigua and is available at the \texttt{R} package \texttt{DAAG}~\cite{Maindonald}. The response measures harvest. We removed the irrelevant \emph{id} variable and one observation with a clearly outlying value of \emph{ears} variable and further worked with a data set consisting of 287 observations with 2 continuous and 3 categorical variables with $p=24$.

\subsection{Setup of experiments}

We evaluated the following methods: \\
- Lasso (\texttt{cv.glmnet} from the \texttt{R} package \texttt{glmnet}), \\
- Group Lasso (\texttt{cv.grpreg} from the \texttt{R} package \texttt{grpreg} with \texttt{penalty="grLasso"}), \\
- Group MCP (\texttt{cv.grpreg} from the \texttt{R} package \texttt{grpreg} with \texttt{penalty="grMCP"} and with \texttt{gamma} set as the default for continuous response, and as 250 for the binary response, which proved to give smaller errors than the default value),\\
- SCOPE from the \texttt{R} package \texttt{CatReg} \citep{stokell2021}. Tuning parameter \texttt{gamma} was chosen as suggested in that paper:  8 or 32 for continuous responses, and as 100 or 250 for the binary responses, \\
- DMR (\texttt{DMRnet} from the \texttt{R} package \texttt{DMRnet}), \\
- GLAMER, which is described above.

All implementation details and codes of above procedures and results of experiments are attached to the paper in the Supplementary Materials. 

 The number of folds in all above-mentioned cross-validation methods was set to 10.

For a given data set and for each evaluated prediction method we performed 100 iterations of the following procedure:
\begin{enumerate}
    \item A training set of $m$ percent of a total number of observations was randomly selected from the data set. 
    \item If there were any variables with only one-level factors, those variables were removed from the training set.
    \item For Insurance, the evaluated methods would repetitively fail to generate trained models, reporting various errors, so this procedure wouldn't terminate. To this end, we additionally removed (greedily) variables resulting in the train data matrix rank deficiency.  
    \item The remaining $100-m$ percent of observations in the data set was treated as a test set for this iteration, with the following modifications: (1) we removed the same variables that had been removed from the train set, (2) we also removed from the test set the observations with levels not present in the training set.
    \item The evaluated method was trained and the resulting model was used to obtain a prediction on a test set. This model and the prediction was then evaluated with the evaluation metrics. If either a training or a test phase couldn't be completed because of reported errors, that iteration was ignored and restarted with a new train sample.
\end{enumerate}
For  Adult and Insurance, we followed~\citet{stokell2021} in choosing $m$ to be equal to 1 and 10, respectively. For Promoter and Antigua data sets (data sets with relatively smaller number of observations) we used $m=70$.

As prediction error (PE) we used mean square error  in the case of continuous responses and misclassification error in the case of binary responses. 

We also measured model dimension (MD) for models reported by each method. The exact calculation depended on an internal representation of a particular model, but the general algorithm could be summarized as follows: sum the number of non-zero parameters, including continuous variables and the intercept (counting the same values related to different levels of the same factor only ones) and subtract the number of constraints imposed in model construction.

\begin{figure}[th]
\centering
\includegraphics[width=1\columnwidth]{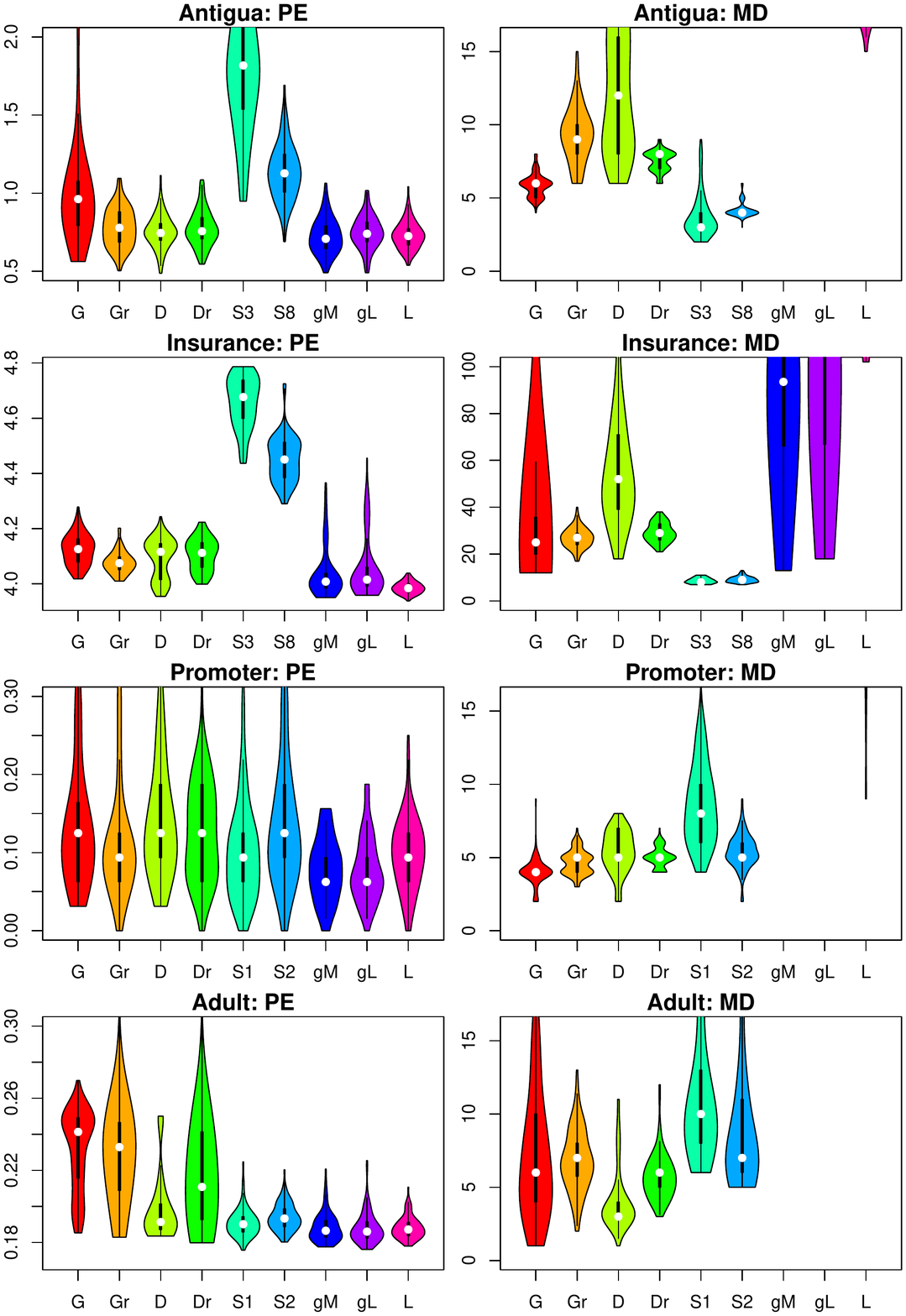} 
\caption{Mean prediction error (PE) and mean model dimension (MD) of considered methods. Details are given in the text.}
\label{fig1}
\end{figure}
 In Figure \ref{fig1} we present mean prediction error (PE) and mean model dimension (MD) of various methods. By default a final model is selected by ten-fold cross-validation. Algorithms used are: G -- GLAMER; Gr -- GLAMER with RIC; D -- DMR; Dr -- DMR with RIC; S3, S8 -- SCOPE-32 and SCOPE-8, respectively, for regression; S1, S2 -- SCOPE-100 and SCOPE-250, respectively, for binary classification; gM -- group MCP; gL -- Group Lasso; L -- Lasso. Most of violins for gM, gL or L are not visible, because their MD are beyond the range of the plot.

\section{Conclusion}

DMR and SCOPE are perhaps the only methods for partition selection for high-dimensional data, but so far there is no proof of selection consistency for either of them in this scenario. We propose GLAMER as a benchmark for theoretical considerations and prove its consistency. Our approach encompasses fundamental models for prediction, that is linear and logistic regression. GLAMER is a simple  generalization of the Thresholded Lasso on partition selection. 

The algorithms for partition selection: GLAMER, DMR and SCOPE have predictive accuracy slightly worse than the Lasso, Group Lasso and Group MCP, but their outputs are significantly more sparse and  have only a few parameters. Therefore, their results are much more interpretable. For data sets with continuous response (Antigua, Insurance), the SCOPE error is slightly larger than errors of the GLAMER and DMR, but its model dimension is smaller. For Adult data, the opposite is true. This difference can probably be reduced by mixing the PE and MD criteria for final selection. 

As a by-product we obtain  optimal weights for the Group Lasso, which are different from those recommended by the authors of this  method. Possibly, the new weights can improve asymptotics of the Group Lasso in the general scenario (not necessarily orthogonal)  and its practical performance as well.

\bibliographystyle{apalike}

\bibliography{refs}

\begin{thebibliography}{}

\bibitem[Andrews and Herzberg, 1985]{Andrews}
Andrews, D.~F. and Herzberg, A.~M. (1985).
\newblock {\em Data. A Collection of Problems from Many Fields for the Student
  and Research Worker}.
\newblock Springer, New York.

\bibitem[Bickel et~al., 2009]{BickelEtAl09}
Bickel, P., Ritov, Y., and Tsybakov, A. (2009).
\newblock Simultaneous analysis of {L}asso and {D}antzig selector.
\newblock {\em Annals of Statistics}, 37:1705--1732.

\bibitem[Blazere et~al., 2014]{Blazere2014}
Blazere, M., Loubes, J.-M., and Gamboa, F. (2014).
\newblock Oracle inequalities for a group lasso procedure applied to
  generalized linear models in high dimension.
\newblock {\em IEEE Transactions on Information Theory}, 60:2303--2318.

\bibitem[Bondell and Reich, 2009]{bondell2009}
Bondell, H.~D. and Reich, B.~J. (2009).
\newblock Simultaneous factor selection and collapsing levels in anova.
\newblock {\em Biometrics}, 65:169--177.

\bibitem[Breheny and Huang, 2015]{BrehenyHuang15}
Breheny, P. and Huang, J. (2015).
\newblock Group descent algorithms for nonconvex penalized linear and logistic
  regression models with grouped predictors.
\newblock {\em Statistics and Computing}, 25:173--187.

\bibitem[B\"uhlmann and van~de Geer, 2011]{BuhlmannGeer11}
B\"uhlmann, P. and van~de Geer, S. (2011).
\newblock {\em Statistics for High-dimensional Data}.
\newblock Springer, New York.

\bibitem[Dua and Graff, 2017]{Dua_Graff}
Dua, D. and Graff, C. (2017).
\newblock {UCI} machine learning repository.
\newblock University of California, Irvine, School of Information and Computer
  Sciences.

\bibitem[Foster and George, 1994]{FosterGeorge1994}
Foster, D. and George, E. (1994).
\newblock The risk inflation criterion for multiple regression.
\newblock {\em Annals of Statistics}, 22:1947--1975.

\bibitem[Friedman et~al., 2010]{FriedmanEtAl10}
Friedman, J., Hastie, T., and Tibshirani, R. (2010).
\newblock Regularization paths for generalized linear models via coordinate
  descent.
\newblock {\em Journal of Statistical Software}, 33:1--22.

\bibitem[Garc{\'\i}a-Donato and Paulo, 2021]{garcia2021variable}
Garc{\'\i}a-Donato, G. and Paulo, R. (2021).
\newblock Variable selection in the presence of factors: a model selection
  perspective.
\newblock {\em Journal of the American Statistical Association}, pages 1--11.

\bibitem[Gertheiss and Tutz, 2010]{gertheiss2010sparse}
Gertheiss, J. and Tutz, G. (2010).
\newblock Sparse modeling of categorial explanatory variables.
\newblock {\em The Annals of Applied Statistics}, pages 2150--2180.

\bibitem[Harley and Reynolds, 1987]{Harley}
Harley, C. and Reynolds, R. (1987).
\newblock Analysis of e. coli promoter sequences.
\newblock {\em Nucleic Acids Research}, 15:2343--2361.

\bibitem[Huang and Zhang, 2012]{HuangZhang12}
Huang, J. and Zhang, C. (2012).
\newblock Estimation and selection via absolute penalized convex minimization
  and its multistage adaptive applications.
\newblock {\em Journal of Machine Learning Research}, 13:1839--1864.

\bibitem[Ibragimov and Has'minskii, 1981]{Ibragimov81}
Ibragimov, I.~A. and Has'minskii, R.~Z. (1981).
\newblock {\em Statistical estimation}.
\newblock Springer, New York.

\bibitem[Kaggle, 2015]{Kaggle}
Kaggle (2015).
\newblock Prudential life insurance assessment.
\newblock https://www.kaggle.com/c/prudential-life-insurance-assessment/data.

\bibitem[Kohavi, 1996]{Kohavi}
Kohavi, R. (1996).
\newblock Scaling up the accuracy of naive-bayes classifiers: A decision-tree
  hybrid.
\newblock In {\em Proceedings of the Second International Conference on
  Knowledge Discovery and Data Mining}, KDD'96, pages 202--207. AAAI Press.

\bibitem[Lounici et~al., 2011]{Lounici_etal2011}
Lounici, K., Pontil, M., van~de Geer, S., and Tsybakov, A.~B. (2011).
\newblock Oracle inequalities and optimal inference under group sparsity.
\newblock {\em The Annals of Statistics}, 39:2164--2204.

\bibitem[Maindonald and Braun, 2010]{Maindonald}
Maindonald, J. and Braun, W. (2010).
\newblock {\em Data Analysis and Graphics Using R}.
\newblock Cambridge University Press.

\bibitem[Maj-Ka\'{n}ska et~al., 2015]{MajK2015}
Maj-Ka\'{n}ska, A., Pokarowski, P., and Prochenka, A. (2015).
\newblock {Delete or merge regressors for linear model selection}.
\newblock {\em Electronic Journal of Statistics}, 9:1749 -- 1778.

\bibitem[Nardi and Rinaldo, 2008]{NardiRinaldo2008}
Nardi, Y. and Rinaldo, A. (2008).
\newblock {On the asymptotic properties of the group lasso estimator for linear
  models}.
\newblock {\em Electronic Journal of Statistics}, 2:605--633.

\bibitem[Oelker et~al., 2014]{oelker14}
Oelker, M.-R., Gertheiss, J., and Tutz, G. (2014).
\newblock Regularization and model selection with categorical predictors and
  effect modifiers in generalized linear models.
\newblock {\em Statistical Modelling}, 14:157--177.

\bibitem[Pauger and Wagner, 2019]{pauger2019bayesian}
Pauger, D. and Wagner, H. (2019).
\newblock Bayesian effect fusion for categorical predictors.
\newblock {\em Bayesian Analysis}, 14(2):341--369.

\bibitem[Prochenka-Soltys and Pokarowski, 2018]{DMRnet}
Prochenka-Soltys, A. and Pokarowski, P. (2018).
\newblock {\em DMRnet: Delete or Merge Regressors Algorithms for Linear and
  Logistic Model Selection and High-Dimensional Data}.
\newblock https://CRAN.R-project.org/package=DMRnet.

\bibitem[Shen et~al., 2013]{ShenEtAl13}
Shen, X., Pan, W., Zhu, Y., and Zhou, H. (2013).
\newblock On constrained and regularized high-dimensional regression.
\newblock {\em Annals of the Institute of Statistical Mathematics},
  65:807--832.

\bibitem[Simon et~al., 2013]{simon2013sparse}
Simon, N., Friedman, J., Hastie, T., and Tibshirani, R. (2013).
\newblock A sparse-group lasso.
\newblock {\em Journal of computational and graphical statistics},
  22(2):231--245.

\bibitem[Stokell et~al., 2021]{stokell2021}
Stokell, B.~G., Shah, R.~D., and Tibshirani, R.~J. (2021).
\newblock Modelling high-dimensional categorical data using nonconvex fusion
  penalties.
\newblock {\em Journal of the Royal Statistical Society: Series B (Statistical
  Methodology)}, 83:579--611.

\bibitem[Tibshirani, 1996]{Tibshirani96}
Tibshirani, R. (1996).
\newblock Regression shrinkage and selection via the {L}asso.
\newblock {\em Journal of the Royal Statistical Society Series B}, 58:267--288.

\bibitem[Tibshirani et~al., 2005]{fused2005}
Tibshirani, R., Saunders, M., Rosset, S., Zhu, J., and Knight, K. (2005).
\newblock Sparsity and smoothness via the fused lasso.
\newblock {\em Journal of the Royal Statistical Society: Series B (Statistical
  Methodology)}, 67:91--108.

\bibitem[Towell et~al., 1990]{Towell}
Towell, G., Shavlik, J., and Noordewier, M. (1990).
\newblock Refinement of approximate domain theories by knowledge-based
  artificial neural networks.
\newblock In {\em Proceedings of the Eighth National Conference on Artificial
  Intelligence (AAAI-90)}. AAAI Press.

\bibitem[van~de Geer and B\"uhlmann, 2009]{GeerBuhlmann09}
van~de Geer, S. and B\"uhlmann, P. (2009).
\newblock On the conditions used to prove oracle results for the {L}asso.
\newblock {\em Electronic Journal of Statistics}, 3:1360--1392.

\bibitem[Ye and Zhang, 2010]{YeZhang10}
Ye, F. and Zhang, C. (2010).
\newblock Rate minimaxity of the {L}asso and {D}antzig {S}elector for the $l_q$
  loss in $l_r$ balls.
\newblock {\em Journal of Machine Learning Research}, 11:3519--3540.

\bibitem[Yuan and Lin, 2006]{groupLasso}
Yuan, M. and Lin, Y. (2006).
\newblock Model selection and estimation in regression with grouped variables.
\newblock {\em Journal of the Royal Statistical Society: Series B (Statistical
  Methodology)}, 68:49--67.

\bibitem[Zhang and Zhang, 2012]{ZhangZhang12}
Zhang, C. and Zhang, T. (2012).
\newblock A general theory of concave regularization for high-dimensional
  sparse estimation problems.
\newblock {\em Statistical Science}, 27:576--593.

\bibitem[Zhou, 2009]{Zhou09}
Zhou, S. (2009).
\newblock Thresholding procedures for high dimensional variable selection and
  statistical estimation.
\newblock In {\em NIPS}, pages 2304--2312.

\end{thebibliography}

\end{document}